\documentclass{article}
\usepackage[utf8]{inputenc}
\usepackage[margin=3cm]{geometry}
\usepackage{charter}
\usepackage{graphicx}
\usepackage{xparse}
\usepackage{amsthm}
\usepackage{amssymb}
\usepackage{amsmath}
\usepackage{authblk}
\usepackage{bbm}
\usepackage{geometry}
\usepackage{mathtools}
\usepackage{enumitem}
\usepackage{tikz}
\usepackage{subcaption}
\usepackage[colorlinks,citecolor=blue,linkcolor=blue]{hyperref}
\usepackage{hypcap} 
\usepackage[capitalise,noabbrev]{cleveref}
\usepackage[textsize=footnotesize]{todonotes}
\presetkeys{todonotes}{inline}{}
\usepackage{thmtools}
\usepackage{thm-restate}
\usepackage{xspace}

\DeclareDocumentCommand\R{}{\mathbb{R}}
\DeclareDocumentCommand\Z{}{\mathbb{Z}}
\DeclareDocumentCommand\conv{o}{\operatorname{conv}\IfValueTF{#1}{\left(#1\right)}{}}
\DeclareDocumentCommand\setdef{mo}{\left\{#1\IfNoValueTF{#2}{}{ \mid #2}\right\}}
\DeclareDocumentCommand\bigsetdef{mo}{\left\{#1\IfNoValueTF{#2}{}{ \mid #2}\right\}}
\DeclareDocumentCommand\transpose{m}{#1^{\intercal}}
\DeclareDocumentCommand\zerovec{o}{\IfNoValueTF{#1}{\mathbb{O}}{\mathbb{O}_{#1}}}
\DeclareDocumentCommand\aff{o}{\operatorname{aff}\IfValueTF{#1}{\left(#1\right)}{}}
\DeclareDocumentCommand\supp{o}{\operatorname{supp}\IfValueTF{#1}{\left(#1\right)}{}}

\DeclareDocumentCommand\Pstab{}{P_{\text{stab}}}
\DeclareDocumentCommand\Redgestab{}{R_{\text{edge}}}
\DeclareDocumentCommand\Rcliquestab{}{R_{\text{clq}}}
\DeclareDocumentCommand\Roddcyclestab{}{R_{\text{oc}}}
\DeclareDocumentCommand\Roddcyclestabfive{}{R_{\text{oc5}}}
\DeclareDocumentCommand\Rstab{}{R}
\DeclareDocumentCommand\graphs{}{\mathcal{G}}
\DeclareDocumentCommand\opt{mm}{\operatorname{opt}(#1,#2)}
\DeclareDocumentCommand\cplxNP{}{\mathsf{NP}}

\newtheorem{theorem}{Theorem}

\newtheorem{corollary}[theorem]{Corollary}
\newtheorem{proposition}[theorem]{Proposition}
\newtheorem{claim}[theorem]{Claim}
\theoremstyle{definition}
\newtheorem{definition}[theorem]{Definition}

\declaretheorem[style=definition]{example}

\tikzstyle{vertex} = [circle,draw,minimum size=17pt,inner sep=0pt]
\tikzstyle{node} = [draw,circle,minimum size=6mm,inner sep=0mm,font=\footnotesize]
\tikzstyle{edge} = [draw,very thick,-]

\newcommand*{\claimproofname}{Proof}

\allowdisplaybreaks

\title{Persistency of Linear Programming Relaxations for the Stable Set Problem}

\author{Elisabeth Rodr{\'i}guez-Heck}
\author{Karl Stickler}
\affil{RWTH Aachen University, Germany}
\author{Matthias Walter}
\affil{University of Twente, The Netherlands}
\author{Stefan Weltge}
\affil{Technical University of Munich, Germany}

\date{\small\today}

\begin{document}

\renewcommand\thmcontinues[1]{\textbf{Continued}}

\maketitle

\begin{abstract}
  The Nemhauser-Trotter theorem states that the standard linear programming (LP) formulation for the stable set problem has a remarkable property, also known as (weak) \emph{persistency}: for every optimal LP solution that assigns integer values to some variables, there exists an optimal integer solution in which these variables retain the same values.
  While the standard LP is defined by only non-negativity and edge constraints, a variety of other LP formulations have been studied and one may wonder whether any of them has this property as well.
  We show that any other formulation that satisfies mild conditions cannot have the persistency property on all graphs, unless it is always equal to the stable set polytope.
\end{abstract}

\section{Introduction}
\label{sec_intro}

Given an undirected graph $ G $ with node set $V(G)$ and edge set $E(G)$, and node weights $ w \in \R^{V(G)} $, the (weighted) stable set problem asks for finding a stable set $ S $ in $ G $ that maximizes $ \sum_{v \in S} w_v $, where a set $ S $ is called stable if $ G $ contains no edge with both endpoints in $ S $.
While the stable set problem is $\cplxNP$-hard, it is a common approach to maximize $ \transpose{w}x $ over the \emph{edge relaxation}
\[
    \Redgestab(G) := \setdef{x \in [0,1]^{V(G)}}[x_v + x_w \le 1 \text{ for each edge } \{v,w\} \in E(G)]
\]
and use optimal (fractional) solutions to gain insights about optimal $ 0/1 $-solutions.
Note that the $ 0/1 $-points in the edge relaxation are precisely the characteristic vectors of stable sets in $ G $, and that maximizing a linear objective over the edge relaxation is a linear program that can be solved efficiently.
Given an optimal solution of this linear program, its objective value is clearly an upper bound on the value of any $ 0/1 $-solution and its entries may guide initial decisions in a branch-and-bound algorithm.
While this is also the case for general polyhedral relaxations, it turns out that optimal solutions of the edge relaxation have a remarkable property that allows us to reduce the size of the problem by fixing some variables to provable optimal integer values.

\begin{definition}[Persistency]
  \label{def_persistency}
  We say that a polytope $ P \subseteq [0,1]^n $ has the \emph{persistency property} if for every objective vector $ c \in \R^n $ and every $ c $-maximal point $ x \in P $, there exists a $ c $-maximal integer point $ y \in P \cap \{0,1\}^n $ such that $ x_i = y_i $ for each $ i \in \{1,2,\dotsc,n\} $ with $ x_i \in \{0,1\} $.
\end{definition}

\begin{proposition}[Nemhauser and Trotter~\cite{NemhauserT75}]
  \label{thm_persistency}
  The edge relaxation $\Redgestab(G)$ has the persistency property for every graph $G$.
\end{proposition}

In other words, the result of Nemhauser and Trotter~\cite{NemhauserT75} states that if $ x^\star $ is an optimal solution for the edge relaxation, then there exists an optimal stable set $ S^\star $ satisfying $ V_1 \subseteq S^\star \subseteq V(G) \setminus V_0 $, where $ V_i := \setdef{v \in V(G)}[x^\star_v = i] $ for $ i = 0,1 $.
In this case, the nodes in $ V_0 \cup V_1 $ can be deleted and the search only has to be performed on the remaining graph.
Clearly, this reduction is significant if $ x^\star $ assigns integer values to many nodes. 

Picard and Queyranne~\cite{Picard77} showed that there is a unique maximal set of variables taking integer values in an optimal solution for the edge relaxation.
For the maximum cardinality stable set problem it has been shown that the probability of obtaining at least a single integer component when solving the LP relaxation is very low for large random graphs~\cite{Pulleyblank79}.
This is shown by proving that almost all graphs $G$ (in the probabilistic sense) are $2$-bicritical, that is, if any node is deleted from $G$, the resulting graph has a $2$-matching, and by proving that $2$-bicritical graphs are precisely those graphs for which any optimal solution of the linear programming relaxation of the uncapacitated $2$-matching problem has no integer components.
By duality of the linear programming relaxations of matching and vertex cover, and by observing that a node set $C \subseteq V$ is a vertex cover of $G = (V, E)$ if and only if $V \setminus C$ is a stable set in $G$, the statement follows for stable set.
Moreover, Grimmett and Pulleyblank~\cite{GrimmettP85,Grimmett86} proved that the probability that certain sparse random graphs are $2$-bicritical, tends to 1 as the number of nodes of the graph increases.
However, the results in~\cite{Pulleyblank79} leave the possibility open, that for different weights or for special classes of graphs, the result of Nemhauser and Trotter~\cite{NemhauserT75} is still of practical interest.
This has indeed been the case in recent years: persistencies have been proved to be very useful in a different context, when dealing with highly structured instances arising in the field of computer vision.
More precisely, Hammer, Hansen and Simeone~\cite{HammerHS84} provided a reduction of (Unconstrained) Quadratic Binary Programming (QBP) to the stable set problem and showed that weak persistency holds for (QBP) as well.
Image restoration problems in computer vision can be formulated as polynomial optimization problems in (millions of) binary variables, which are then reformulated as a (QBP) problem with the motivation of using persistencies to reduce the problem size.
Boros et al.~\cite{BorosHST08} provided an algorithm to compute the largest possible set of variables to fix via persistencies in a quadratic binary program in polynomial time, which has been successfully used in practice to solve very large image restoration problems~\cite{KolmogorovR07,Ishikawa11,FixGBZ15}.

In general, dual bounds obtained from the edge relaxation are quite weak, and several families of additional inequalities have been studied in order to strengthen this formulation.
Examples are the clique inequalities~\cite{Padberg73}, (lifted) odd-cycle inequalities~\cite{Padberg73,Trotter75} and clique-family inequalities~\cite{Oriolo03}.
Most of these families were discovered by systematically studying the facets of the \emph{stable set polytope} $\Pstab(G)$, which is the convex hull of the characteristic vectors of stable sets in $G$.
The stable set polytope itself is known to be a complicated polytope.
In particular, one cannot expect to be able to completely characterize its facial structure~\cite{KarpP82}.
Thus, the following question is natural.

\begin{center}
  \begin{minipage}{0.9\textwidth}
    \emph{Do there exist other linear programming formulations for the stable set problem that also have the persistency property for every graph $G$?}
  \end{minipage}
\end{center}
In this paper, we answer the question negatively.
More precisely, we show that an LP formulation (satisfying mild conditions) different from the edge formulation cannot have the persistency property on all graphs, unless it always yields the stable set polytope.

\paragraph{Outline.}
The paper is structured as follows.
We start by introducing the conditions we impose on the LP formulation in \cref{sec_families}.
Our main result and its consequences are presented in \cref{sec_results}.
\cref{sec_proof} is dedicated to the proof of the main result.
In \cref{sec_necessity_conditions} we show that all requirements for our main result are indeed necessary.
The paper is concluded in \cref{sec_conclusion}, where we discuss open problems.

\section{LP formulations for stable set}
\label{sec_families}

It is clear that, when considering a \emph{single} graph $ G $, one can artificially construct polytopes different from $\Pstab(G)$ and $\Redgestab(G)$ 
that have the persistency property. For instance, if $x \in \R^{V(G)}$ is any point that has only fractional coordinates such that 
$Q := \conv(\Pstab(G) \cup \{x\})$ satisfies $Q \cap \Z^{V(G)} = \Pstab(G) \cap \Z^{V(G)}$, then $Q$ has the persistency property for trivial reasons. In this work, however, we consider relaxations defined for \emph{every} graph that arise in a more structured way.

To this end, let $ \graphs $ denote the set of finite undirected simple graphs.
We regard an LP \emph{formulation} for the stable set problem as a map that assigns to every graph $ G \in \graphs $ a polytope $ \Rstab(G) $ such that $\Rstab(G) \cap \Z^{V(G)} = \Pstab(G) \cap \Z^{V(G)} $ holds.
As an example, the edge formulation assigns $ \Redgestab(G) $ to every graph $ G $.
Next, let us specify some natural conditions that are satisfied by many prominent formulations and under which our main result holds.
Each of these conditions is defined for a formulation $ \Rstab $.

\paragraph{Condition (A).}
The inequalities defining $\Rstab$ are derived from facets of $\Pstab$.
Formally,

\begin{equation*}
  \text{\begin{minipage}{0.85\textwidth}
    for each $G \in \graphs$, each inequality with support $U \subseteq V(G)$ that is facet-defining for $\Rstab(G)$ is also facet-defining for $\Pstab(G[U])$,
  \end{minipage}}
  \tag{A}
  \label{eq_property_facet}
\end{equation*}
where $G[U]$ denotes the subgraph induced by $U$.
Note that inequalities need to define facets only on their support graph.
For example, odd-cycle inequalities satisfy~\eqref{eq_property_facet} although in general they do not define facets~\cite{Padberg73}.
By construction, also lifted odd-cycle inequalities (see~\cite{Trotter75}) satisfy~\eqref{eq_property_facet}.
However, our results do not cover a formulation consisting of only rank inequalities (see~\cite{Chvatal75}) since the latter are not facet-defining in general.

\paragraph{Condition (B).}

For every graph $G \in \graphs$, validity of facet-defining inequalities of $\Rstab(G)$ is inherited by induced subgraphs. Formally,

\begin{equation*}
  \text{\begin{minipage}{0.85\textwidth}
    for each $G \in \graphs$, each inequality with support $U \subseteq V(G)$ that is facet-defining for $\Rstab(G)$ is valid (although not necessarily facet-defining) for $\Rstab(G[U])$.
  \end{minipage}}
  \tag{B}
  \label{eq_property_valid}
\end{equation*}
This requirement ensures that if an (irredundant) inequality arises for some graph then it must (at least implicitly) occur for all induced subgraphs for which it is defined. The reverse implication is imposed by the third condition, although in a more structured way.
For this, we need the following definitions.

\DeclareDocumentCommand\onesum{mm}{\oplus^{#1}_{#2}}

Let $ G_1, G_2 \in \graphs $ and let $ v_1 \in V(G_1)$, $ v_2 \in V(G_2)$.
Then the \emph{$1$-sum of $G_1$ and $G_2$ at $v_1$ and $v_2$}, denoted by $G_1 \onesum{v_1}{v_2} G_2$ is the graph obtained from the disjoint union of $G_1$ and $G_2$ by identifying $v_1$ with $v_2$.
Moreover, let $P \subseteq \R^m$ and $Q \subseteq \R^n$ be polytopes and let $i \in \{1,2,\dotsc,m\}$ and $j \in \{1,2,\dotsc,n\}$.
The \emph{$1$-sum of $P$ and $Q$ at coordinates $i$ and $j$}, denoted by $P \onesum{i}{j} Q$, is defined as the projection of $\conv( \setdef{ (x,y) \in P \times Q }[ x_i = y_j ] )$ onto all variables except for $y_j$.
Notice that this projection is an isomorphism from the convex hull to its image since the variables $x_i$ and $y_j$ are equal.

\paragraph{Condition (C).}
For every pair of graphs $G_1, G_2 \in \graphs$, validity of inequalities is acquired by their 1-sum. Formally,

\begin{equation*}
  \text{\begin{minipage}{0.85\textwidth}
    $\Rstab(G_1 \onesum{v_1}{v_2} G_2) \subseteq \Rstab(G_1) \onesum{v_1}{v_2} \Rstab(G_2)$ holds for all $G_1,G_2 \in \graphs$ and all nodes $v_1 \in V(G_1)$ and $v_2 \in V(G_2)$.
  \end{minipage}}
  \tag{C}
  \label{eq_property_onesum}
\end{equation*}
Also this condition is very natural since every inequality that is valid for $\Rstab(G_1)$ is also valid for $\Pstab(G_1 \onesum{v_1}{v_2} G_2)$, and hence its participation in $\Rstab(G_1 \onesum{v_1}{v_2} G_2)$ is reasonable.

Before we state our main result, let us mention immediate observations, which are summarized below.

\begin{proposition}
  \label{thm_summary_conditions}
  \leavevmode
  \begin{enumerate}[label=(\roman*)]
    \item
      \label{thm_summary_conditions_edgestab_stab}
      $\Redgestab$ and $\Pstab$ satisfy~\eqref{eq_property_facet}--\eqref{eq_property_onesum}.
    \item
      \label{thm_summary_conditions_reverse}
      Let $ \Rstab $ be any stable set formulation satisfying~\eqref{eq_property_facet} and~\eqref{eq_property_valid}.
      Then $\Rstab(G_1 \onesum{v_1}{v_2} G_2) \supseteq \Rstab(G_1) \onesum{v_1}{v_2} \Rstab(G_2)$ holds for all $G_1,G_2 \in \graphs$ and all nodes $v_1 \in V(G_1)$ and $v_2 \in V(G_2)$.
    \item
      \label{thm_summary_conditions_closed}
      If $ \Rstab^1 $ and $ \Rstab^2 $ satisfy~\eqref{eq_property_facet}--\eqref{eq_property_onesum}, then $\{\Rstab^1(G) \cap \Rstab^2(G)\}_{G \in \graphs}$ also satisfies~\eqref{eq_property_facet}--\eqref{eq_property_onesum}.
  \end{enumerate}
\end{proposition}

In other words, the second observation states that the reverse inclusion of~\eqref{eq_property_onesum} is implied by~\eqref{eq_property_facet} and~\eqref{eq_property_valid}, and the last observation is that our conditions are closed under intersection of relaxations.

\begin{proof}
  It is clear that $\Redgestab$ satisfies~\eqref{eq_property_facet}--\eqref{eq_property_onesum}, and that $ \Pstab $ satisfies~\eqref{eq_property_valid}.
  Chv{\'a}tal proved Property~\eqref{eq_property_onesum} for $\Pstab$ by showing that the stable set polytope of a clique-sum of two graphs is obtained from the stable set polytopes of these two graphs without adding inequalities (see Theorem~4.1 in~\cite{Chvatal75}).
  To see Property~\eqref{eq_property_facet} for $\Pstab$, observe that the stable set polytope of an induced subgraph is isomorphic to a face defined by the nonnegativity constraints of the removed nodes.
  This shows~\ref{thm_summary_conditions_edgestab_stab}.

  To see~\ref{thm_summary_conditions_reverse}, let $G = G_1 \onesum{v_1}{v_2} G_2$, and consider an inequality that is facet-defining for $\Rstab(G)$ and has support on $U \subseteq V(G)$.
  By~\eqref{eq_property_facet}, it is facet-defining for $\Pstab(G[U])$.
  By \ref{thm_summary_conditions_edgestab_stab}, $\Pstab(G[U])$ satisfies~\eqref{eq_property_onesum}, that is, the support $U$ has to satisfy $U \subseteq V(G_1)$ or $U \subseteq V(G_2)$.
  By~\eqref{eq_property_valid}, the inequality must be valid for $\Rstab(G_1)$ or $\Rstab(G_2)$, which concludes the proof.

  For~\ref{thm_summary_conditions_closed}, all three properties can be shown by inspection of individual inequalities of $\Rstab^1$ and $\Rstab^2$.
\end{proof}

Many LP formulations that are stronger than the edge formulation satisfy Conditions~\eqref{eq_property_facet}--\eqref{eq_property_onesum}.
However, there even exist formulations $\Rstab$ with $\Rstab(G) \not\subseteq \Redgestab(G)$ for some graphs $G$ that also satisfy the three conditions.
We illustrate for a particular graph how such a formulation can be constructed.
To this end, consider an odd wheel $W_5$ consisting of an odd cycle of length $5$ with nodes $1,2,3,4,5$ plus a center node $6$.
The formulation for this graph is defined as follows:
\begin{multline*}
 \Rstab(W_5) \coloneqq \{ x \in \R^6 \mid x_i \geq 0 \text{ for } i=1,2,3,4,5,6 \text{, } x_1 + x_2 + x_3 + x_4 + x_5 + 2x_6 \leq 2, \\ x_1 + x_2 \leq 1 \text{, } x_2 + x_3 \leq 1 \text{, } x_3 + x_4 \leq 1 \text{, } x_4 + x_5 \leq 1 \text{ and } x_5 + x_1 \leq 1 \}.
\end{multline*}
It is easy to see that $2x_i + 2x_6 \leq 3$ is valid for $\Rstab(W_5)$ for $i=1,2,3,4,5$, which establishes $\Rstab(W_5) \cap \Z^6 = \Pstab(W_5) \cap \Z^6$.
Moreover, since $\transpose{(1,0,0,0,0,0.5)} \in \Rstab(W_5)$, the edge inequality $x_1 + x_6 \leq 1$ is not valid for $\Rstab(W_5)$ and thus $\Rstab(W_5) \not\subseteq \Redgestab(W_5)$.

\section{Results}
\label{sec_results}

We say that two formulations $ \Rstab^1 $ and $ \Rstab^2 $ are \emph{equivalent} if $ \Rstab^1(G) = \Rstab^2(G) $ holds for every $ G \in \graphs $, in which case we write $ \Rstab^1 \equiv \Rstab^2 $. We can now state our main result.

\begin{theorem}
  \label{thm_main}
  Let $\Rstab$ be a stable set formulation satisfying~\eqref{eq_property_facet}--\eqref{eq_property_onesum}.
  Then $\Rstab(G)$ has the persistency property for all graphs $ G \in \graphs$ if and only if $\Rstab \equiv \Redgestab$ or $\Rstab \equiv \Pstab$.
\end{theorem}

Sufficiency follows from \cref{thm_persistency} and from the fact that $\Pstab(G)$ is an integral polytope for every $ G \in \graphs $.
Before we prove necessity in \cref{sec_proof}, let us mention some direct implications of Theorem~\ref{thm_main} for known relaxations.

\begin{corollary}
  \label{thm_clique}
  The clique relaxation
  \begin{equation*}
    \Rcliquestab(G) = \bigsetdef{ x \in \Redgestab(G) }[ x(V(C)) \leq 1 \text{ for each clique $C$ of $G$ } ]
  \end{equation*}
  does not have the persistency property for all graphs $G \in \graphs$.
\end{corollary}

\begin{proof}
  It is easy to see that $\Rcliquestab$ satisfies Property~\eqref{eq_property_onesum}.
  For Properties~\eqref{eq_property_facet} and~\eqref{eq_property_valid}, consider a clique $C$ of some graph $G \in \graphs$.
  Clearly, $C$ is also a clique of $G[V(C)]$ and the inequality is known to be facet-defining for $\Pstab(G[V(C)])$ (see Theorem~2.4 in~\cite{Padberg73}).
\end{proof}

Also the relaxation based on odd-cycle inequalities satisfies these properties, although the inequalities are generally not facet-defining.

\begin{corollary}
  \label{thm_oddcycle}
  The odd-cycle relaxation
  \begin{equation*}
    \Roddcyclestab(G) = \bigsetdef{ x \in \Redgestab(G) }[ x(V(C)) \leq \frac{|V(C)|-1}{2} \text{ for each chordless odd cycle $C$ of $G$ } ]
  \end{equation*}
  does not have the persistency property for all graphs $G \in \graphs$.
\end{corollary}

\begin{proof}
  It is easy to see that $\Roddcyclestab$ satisfies Property~\eqref{eq_property_onesum}.
  For Properties~\eqref{eq_property_facet} and~\eqref{eq_property_valid}, consider a chordless odd cycle $C$ of some graph $G \in \graphs$.
  Clearly, $C$ is also a chordless odd cycle of $G[V(C)]$, and the odd-cycle inequality is facet-defining for $\Pstab(G[V(C)])$ (see Theorem~3.3 in~\cite{Padberg73}).
\end{proof}

Using \cref{thm_summary_conditions}~\ref{thm_summary_conditions_closed}, we obtain the same result for their intersection.

\begin{corollary}
  \label{thm_clique_oddcycle}
  The intersection of the clique and the odd-cycle relaxations
  \begin{equation*}
    \Rstab(G) = \Rcliquestab(G) \cap \Roddcyclestab(G)
  \end{equation*}
  does not have the persistency property for all graphs $G \in \graphs$.
\end{corollary}

\paragraph{Strong persistency.}
Hammer, Hansen and Simeone~\cite{HammerHS82} considered a variant of the persistency property that considers coordinates that are fixed to the same integer for all optimal solutions. For every graph $ G $ and objective vector $c \in \R^{V(G)} $ they showed the following. If there is a node $ i \in V(G) $ together with a value $ b \in \{0,1\} $ for which every $ c $-maximal solution $x \in \Redgestab(G)$ satisfies $x_i = b$, then also every $c$-maximal solution $y \in \Pstab \cap \{0,1\}^{V(G)}$ satisfies $y_i = b$. In the pseudo-Boolean optimization literature this is also referred as the \emph{strong persistency} property of the edge relaxation. The necessity proof of \cref{thm_main} will show that our main result also holds for this notion of persistency.

\paragraph{Vertex cover.}
A vertex cover in a graph $G = (V,E)$ is a set $C \subseteq V$ such that $C$ contains at least one endnode of each edge of $G$.
Clearly, $C$ is a vertex cover in $G$ if and only if $V \setminus C$ is a stable set in $G$.
One consequence of this observation is that the map $\pi : V \to V$ defined via $\pi(x)_v = 1 - x_v$ for each $v \in V$ maps $\Pstab(G)$ to the \emph{vertex cover polytope}, which is defined as the convex hull of characteristic vectors of vertex covers in $G$.
Similarly, a natural linear programming relaxation for vertex cover is $\pi(\Redgestab(G))$, which can for instance be strengthened by inequalities that correspond to cliques or odd cycles.
Since also persistency is maintained under the map $\pi$, all our results also hold for vertex cover.

\section{Proof of the main result}
\label{sec_proof}

Let us fix any stable set formulation $\Rstab$ over $\graphs$ satisfying Properties~\eqref{eq_property_facet}--\eqref{eq_property_onesum}.
To prove the ``only if'' implication of \cref{thm_main} we have to verify that if $\Rstab \not\equiv \Redgestab$ and $\Rstab \not\equiv \Pstab$, then $\Rstab(G)$ does not have the persistency property for all graphs $ G \in \graphs$.
Equivalently, we have to prove the following:
\begin{equation*}
  \text{\begin{minipage}{0.85\textwidth}
    If there exist graphs $G_1,G_2 \in \graphs$ with $\Rstab(G_1) \neq \Redgestab(G_1)$ and $\Rstab(G_2) \neq \Pstab(G_2)$, then there exists a graph $G^\star$ for which the polytope $\Rstab(G^\star)$ does not have the persistency property.
  \end{minipage}}
  \tag{$\diamondsuit$}
  \label{eq_only_if}
\end{equation*}
Given $ G_1 $ and $ G_2 $, we will provide an explicit construction of $ G^\star $ and show that $ \Rstab(G^\star) $ does not have the persistency property.
To see the latter, we will give an objective vector $c^\star \in \R^{V(G^\star)}$ such that every $c^\star$-maximal solution over $\Rstab(G^\star)$ has a certain coordinate equal to zero while every $c^\star$-maximal stable set in $G^\star$ contains the corresponding node.

\DeclareDocumentCommand\Gouter{}{G^{\text{out}}}
\DeclareDocumentCommand\vouter{}{v^{\text{out}}}
\DeclareDocumentCommand\couter{}{c^{\text{out}}}
\DeclareDocumentCommand\Ginner{}{G^{\text{in}}}
\DeclareDocumentCommand\GinnerEdge{}{G^{\text{in}}_{\text{edge}}}

The graph $G^\star$ will consist of an ``inner'' graph $\Ginner$ with $\Rstab(\Ginner) \neq \Redgestab(\Ginner)$ and $|V(\Ginner)| - 1$ copies of an ``outer'' graph $\Gouter$ with $\Rstab(\Gouter) \neq \Pstab(\Gouter)$.
Each copy of $\Gouter$ is attached to a node of $\Ginner$ via the 1-sum operation.
The only node of $\Ginner$ that does \emph{not} have a copy of $\Gouter$ attached corresponds precisely to the coordinate showing that $ \Rstab(G^\star) $ does not have the persistency property.
Note that such graphs $\Ginner, \Gouter$ exist due to the hypothesis of~\eqref{eq_only_if}. Among all such graphs, we will make particular choices satisfying some additional properties that we specify in the next sections.

We will illustrate our definitions and the steps of the proof by providing two running examples.

\begin{example}[label=ex_C5_K3]
  Consider the formulation $\Roddcyclestabfive$ defined via
  \begin{multline*}
    \Roddcyclestabfive(G) = \big\{ x \in \Redgestab(G) \mid x(V(C)) \leq \frac{ |V(C)|-1 }{2} \\ \text{ for each chordless odd cycle $C$ of $G$ with at least $5$ nodes } \big\}.
  \end{multline*}
  The hypothesis of~\eqref{eq_only_if} is satisfied for $\Roddcyclestabfive$ because the odd cycle $C_5$ is such that $\Roddcyclestab(C_5) \neq \Redgestab(C_5)$ and the complete graph $K_3$ is such that $\Roddcyclestabfive(K_3) \neq \Pstab(K_3)$.
\end{example}

\begin{example}[label=ex_C3_K4]
  Consider the odd-cycle formulation $\Roddcyclestab$.
  The hypothesis of~\eqref{eq_only_if} is satisfied for $\Roddcyclestab$ because the odd cycle $C_3$ is such that $\Roddcyclestab(C_3) \neq \Redgestab(C_3)$ and the complete graph $K_4$ is such that $\Roddcyclestab(K_4) \neq \Pstab(K_4)$.
\end{example}

\subsection{The graph \texorpdfstring{$\Gouter$}{G-out}}

In the definition of the auxiliary graph $\Gouter$ we will make use of the following lemma. In what follows, for a polytope $P \subseteq \R^n$ and a vector $c \in \R^n$, let us denote the optimal face of $P$ induced by $c$ by $\opt{P}{c} \coloneqq \arg\max\setdef{ \transpose{c}x }[ x \in P]$.

\begin{restatable}{lemma}{rsfacevertex}
  \label{thm_face_vertex}
  Let $P,Q \subseteq \R^n$ be polytopes.
  If there exists a vector $c \in \R^n$ such that $\dim(\opt{Q}{c}) < \dim(\opt{P}{c})$, then there exists a vector $c' \in \R^n$ such that $\opt{Q}{c'}$ is a vertex of $Q$, while $\opt{P}{c'}$ is not a vertex of $P$.
\end{restatable}
\begin{proof}
  See \cref{appendix_proofs}.
\end{proof}

The graph $ \Gouter $ is now defined through the following statement.

\begin{claim}
  \label{thm_bad_graph}
  Assuming the hypothesis of~\eqref{eq_only_if}, there exists a graph $\Gouter \in \graphs$, a vector $\couter \in \R^{V(\Gouter)}$ and a node $\vouter \in V(\Gouter)$ such that $\opt{\Rstab(\Gouter)}{\couter} = \setdef{ \hat{x} }$ holds with $\hat{x}_{\vouter} \geq \tfrac{1}{2}$ and such that $\opt{\Pstab(G)}{\couter}$ contains a vertex $\bar{x} \in \setdef{0,1}^{V(\Gouter)}$ with $\bar{x}_{\vouter} = 0$.
\end{claim}

\begin{proof}
  Let $G \in \graphs$ be such that $\Rstab(G) \neq \Pstab(G)$.
  Such a graph exists by hypothesis of \eqref{eq_only_if}.
  Since $\Rstab(G)$ is a stable set formulation, there exists an inequality $\transpose{a}x \leq \delta$ that is facet-defining for $\Pstab(G)$, but not valid for $\Rstab(G)$.

  We claim that the face $\opt{\Rstab(G)}{a}$ is not a facet of $\Rstab(G)$.
  Assume for a contradiction that $\opt{\Rstab(G)}{a}$ is a facet of $\Rstab(G)$ and define $\delta' \coloneqq \max \setdef{ \transpose{a}x }[ x \in \Rstab(G) ]$.
  Since $\transpose{a}x \leq \delta$ is not valid for $\Rstab(G)$, we have $\delta' > \delta$.
  Property~\eqref{eq_property_facet} implies that $\transpose{a}x \leq \delta'$ is facet-defining for $\Pstab(G[\supp(a)])$, and in particular, equality holds for the characteristic vector of some stable set $S \subseteq V(G[\supp(a)])$.
  Since $S$ is also a stable set in $G$, this contradicts the assumption that $\transpose{a}x \leq \delta$ is valid for $\Pstab(G)$.

  By \cref{thm_face_vertex}, there exists a vector $c \in \R^n$ such that $\opt{\Rstab(G)}{c} = \setdef{\hat{x}}$ and $\opt{\Pstab(G)}{c}$ has (at least) two vertices $\bar{x}^1, \bar{x}^2 \in \setdef{0,1}^{V(G)}$.
  Since $\bar{x}^1 \neq \bar{x}^2$, there exists a coordinate $u \in V(G)$ at which they differ and we can assume $\bar{x}^1_u = 0$ and $\bar{x}^2_u = 1$ without loss of generality.
  If $\hat{x}_u \geq \tfrac{1}{2}$, we can choose $\Gouter \coloneqq G$, $\couter \coloneqq c$ and $\vouter \coloneqq u$.
  Together with $\hat{x}$ and $\bar{x}^1$, they satisfy the requirements of the lemma.

  Otherwise, let $G'$ be the graph $G$ with an additional edge $\setdef{u,u'}$ attached at $u$.
  Formally, let $G''$ be the graph consisting of a single edge $\setdef{u,u'}$ and let $G' \coloneqq G \onesum{u}{u} G''$.
  By Property~\eqref{eq_property_onesum} and \cref{thm_summary_conditions}~\ref{thm_summary_conditions_reverse}, $\Rstab(G') = \Rstab(G) \onesum{u}{u} \Rstab(G'')$ holds.
  Since $G''$ is a single edge, Property~\eqref{eq_property_facet} together with the fact that $\Rstab(G'') \cap \Z^{E(G'')}$ contains only the stable set vectors implies $\Rstab(G'') = \Pstab(G'')$.
  Thus, $\Rstab(G')$ is described by all inequalities that are valid for $\Rstab(G)$ together with $x_{u'} \geq 0$ and $x_u + x_{u'} \leq 1$.
  Hence, for a sufficiently small $\varepsilon > 0$ and the objective vector $c' \in \R^{V(G')}$ with $c'_{u'} = \varepsilon$, $c'_u = c_u + 2\varepsilon$ and $c'_v = c_v$ for all $v \in V(G) \setminus \setdef{u}$, the maximization of $c'$ over $\Rstab(G')$ yields a unique optimum $\hat{x}' \in \R^{V(G')}$ with $\hat{x}'_v = \hat{x}_v$ for all $v \in V(G)$ and $\hat{x}'_{u'} = 1-\hat{x}'_u > \tfrac{1}{2}$, while the maximization of $c'$ over $\Pstab(G')$ admits an optimum $\bar{x}' \in \R^{V(G')}$ with $\bar{x}'_u = 1$ and $\bar{x}'_{u'} = 0$.
  Now, $\Gouter \coloneqq G'$, $\couter \coloneqq c'$ and $\vouter \coloneqq u'$ together with $\hat{x}'$ and $\bar{x}'$ satisfy the requirements of the lemma.
\end{proof}

\begin{example}[continues=ex_C5_K3]
  For $\Roddcyclestabfive$, choose $G$ in the proof of \cref{thm_bad_graph} to be $K_3$, the complete graph on nodes $\{A, B, C\}$.
  We assume that vectors in $\R^{V(G)}$ are indexed in the order $A$, $B$, $C$.
  The clique inequality $x_A + x_B + x_C \leq 1$ is facet-defining for $\Pstab(G)$ but not valid for $\Roddcyclestabfive(G)$.
  Moreover, for $c = \transpose{(1,1,1)}$, $\opt{\Roddcyclestab(G)}{c} = \{\hat{x}\}$ with $ \hat{x} = \transpose{(\frac{1}{2}, \frac{1}{2}, \frac{1}{2})}$ while $\opt{\Pstab(G)}{c}$ contains the three stable sets defined by selecting a single node in $G$, and hence has dimension $2$.
  Consider $\bar{x}_1 = \transpose{(0,1,0)}$ and $\bar{x}_2 = \transpose{(1,0,0)}$ and choose $u = A$.
  Following the proof of \cref{thm_bad_graph}, we obtain $\Gouter = G$, $\vouter = A$ and $\couter = c$ as depicted in \cref{fig_C5_K3_Gout}.
\end{example}

\begin{example}[continues=ex_C3_K4]
  For $\Roddcyclestab$, choose $G$ in the proof of \cref{thm_bad_graph} to be $K_4$, the complete graph on nodes $\{A, B, C, D\}$.
  We assume that vectors in $\R^{V(G)}$ are indexed in the order $A$, $B$, $C$, $D$.
  The clique inequality $x_A + x_B + x_C + x_D \leq 1$ is facet-defining for $\Pstab(G)$ but not valid for $\Roddcyclestab(G)$.
  Moreover, for $c = \transpose{(1,1,1,1)}$, $\opt{\Roddcyclestab(G)}{c} = \{\hat{x}\}$ with $ \hat{x} = \transpose{(\frac{1}{3}, \frac{1}{3}, \frac{1}{3}, \frac{1}{3})}$ while $\opt{\Pstab(G)}{c}$ contains the four stable sets defined by selecting a single node in $G$, and hence has dimension $3$.
  Consider $\bar{x}_1 = \transpose{(0,1,0,0)}$ and $\bar{x}_2 = \transpose{(1,0,0,0)}$ and choose $u = A$.
  Since $\hat{x}_u < \frac{1}{2}$, we introduce an additional edge $\{u, u'\} = \{A, A'\}$ and the graph $G''$ consisting of these two nodes and the single edge.
  Following the proof of \cref{thm_bad_graph}, $\Gouter = G \onesum{A}{A} G''$ and $\vouter = A'$.
  Finally, $\couter$ can be defined by setting $\varepsilon = \frac{1}{3}$.
  \cref{fig_C3_K4_Gout} illustrates $\Gouter$ and $\couter$.
  The unique optimum of maximizing $\couter$ over $\Roddcyclestab(\Gouter)$ is $\hat{x}_v = \frac{1}{3}$ for $v \in \{A,B,C,D\}$ and $\hat{x}_{A'} = \frac{2}{3}$, while selecting only node $A$ is an optimal stable set for $\couter$.
\end{example}

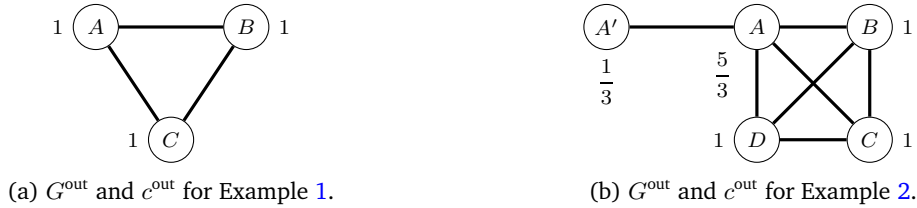
\begin{figure}[htb]
  \begin{subfigure}[b]{0.49\textwidth}
    \centering
    \begin{tikzpicture}
      \node[node,label=left:\footnotesize{$1$}] (A) at (0,0) {$A$};
      \node[node,label=right:\footnotesize{$1$}] (B) at (2,0) {$B$};
      \node[node,label=left:\footnotesize{$1$}] (C) at (1,-1.5) {$C$};

      \draw[edge] (A) -- (B) -- (C) -- (A);
    \end{tikzpicture}
    \caption{$\Gouter$ and $\couter$ for \cref{ex_C5_K3}.}
    \label{fig_C5_K3_Gout}
  \end{subfigure}
  \begin{subfigure}[b]{0.49\textwidth}
    \centering
    \begin{tikzpicture}
      \node[node,label=below:\footnotesize{$\dfrac{1}{3}$}] (Ap) at (-2,0) {$A'$};
      \node[node,label=below left:\footnotesize{$\dfrac{5}{3}$}] (A) at (0,0) {$A$};
      \node[node,label=right:\footnotesize{$1$}] (B) at (1.5,0) {$B$};
      \node[node,label=right:\footnotesize{$1$}] (C) at (1.5,-1.5) {$C$};
      \node[node,label=left:\footnotesize{$1$}] (D) at (0,-1.5) {$D$};

      \draw[edge] (A) -- (B) -- (C) -- (D) -- (A) -- (Ap);
      \draw[edge] (A) -- (C);
      \draw[edge] (B) -- (D);
    \end{tikzpicture}
    \caption{$\Gouter$ and $\couter$ for \cref{ex_C3_K4}.}
    \label{fig_C3_K4_Gout}
  \end{subfigure}
  \caption{Illustrations of graph $\Gouter$ and objective vector $\couter$ for the two running examples.}
\end{figure}

\subsection{The graph \texorpdfstring{$\Ginner$}{G-in}}

Among all graphs $ G \in \graphs$ with $\Rstab(G) \neq \Redgestab(G)$ we choose $\Ginner$ to have a minimum number of nodes.
Note that $\Ginner$ exists by hypothesis of \eqref{eq_only_if}.
We assume $V(\Ginner) = \setdef{1,2,\dotsc,n}$.
In general $\Rstab(\Ginner)$ might not necessarily satisfy all edge inequalities (e.g., for the example at the end of \cref{sec_families}).
To construct our counterexample we consider $\Ginner$, and we define as $\GinnerEdge$ the subgraph of $\Ginner$ having only those edges for which the corresponding edge inequalities are valid for $\Rstab(\Ginner)$.
Formally, let $V(\GinnerEdge) \coloneqq V(\Ginner)$ and $E(\GinnerEdge) \coloneqq \{ \{u,v\} \in E(\Ginner) \mid x_u + x_v \leq 1 \text{ is valid for $\Rstab(\Ginner)$} \}$.
Let $Ax \leq b$ (with $A \in \Z^{m \times n}$ and $b \in \Z^m$) be the system containing inequalities for all facets of $\Rstab(\Ginner)$ that are not valid for $\Redgestab(\Ginner)$.
This implies $\Rstab(\Ginner) = \{ x \in \Redgestab(\GinnerEdge) \mid Ax \leq b \}$.

Note that $n \geq 3$ holds by assumption on $\Ginner$.
Assume, for the sake of contradiction, that $m = 0 $ holds.
Then $\Rstab(\Ginner)$ would only consist of nonnegativity and edge constraints.
This in turn would mean $\Redgestab(\GinnerEdge) = \Rstab(\Ginner) \neq \Redgestab(\Ginner)$ by choice of $\Ginner$, and hence $\GinnerEdge$ would be a proper subgraph of $\Ginner$.
However, due to the absence of additional inequalities, this would imply $\Redgestab(\GinnerEdge) \cap \Z^{V(\Ginner)} = \Rstab(\Ginner) \cap \Z^{V(\Ginner)} \neq \Pstab(\Ginner) \cap \Z^{V(\Ginner)}$, a contradiction.
We conclude that also $m \geq 1$ holds.

\begin{claim}
  \label{thm_full_support}
  $A_{i,j} \geq 1$ holds for every $i \in \{1,2,\dotsc,m\}$ and every $j \in \{1,2,\dotsc,n\}$.
\end{claim}

\begin{proof}
  It is a basic fact that every facet-defining inequality of a stable set polytope that is not a nonnegativity constraint is of the form $\transpose{a}x \leq \beta$ for some nonnegative vector $a \in \R^n$ (see Section~9.3 in~\cite{Schrijver86}).
  Assume, $A_{i,j} = 0$ holds for some indices $i,j$.
  By Property~\eqref{eq_property_valid}, $A_{i,\star}x \leq b_i$ is valid for $\Rstab(\Ginner[\supp(A_{i,\star})])$, while it is not valid for $\Redgestab(\Ginner[\supp(A_{i,\star})])$.
  This contradicts the minimality assumption for $\Ginner$.
\end{proof}

For both our examples, \cref{thm_full_support} is easy to verify.

\begin{example}[continues=ex_C5_K3]
  For $\Roddcyclestabfive$, choose $\Ginner = C_5$ to be a cycle on nodes $\{1,2,3,4,5\}$.
  Then $\GinnerEdge = \Ginner$ and the system $Ax \leq b$ consists of just the odd-cycle inequality $x_1 + x_2 + x_3 + x_4 + x_5 \leq 2$.
\end{example}

\begin{example}[continues=ex_C3_K4]
  For $\Roddcyclestab$, choose $\Ginner = C_3$ to be a triangle on nodes $\{1,2,3\}$.
  Then $\GinnerEdge = \Ginner$ and the system $Ax \leq b$ consists of just the triangle inequality $x_1 + x_2 + x_3 \leq 1$.
\end{example}

\subsection{The graph \texorpdfstring{$G^\star$}{G*}}

For each $j \in \{2,3,\dots,n\}$ let $G^j$ be an isomorphic copy of $\Gouter$ such that $V(G^j) \cap V(G^k) = \varnothing$ whenever $j \neq k$.
Let $c^j \in \R^{V(G^j)}$ and $v^j \in V(G^j)$ be the vector and node corresponding to $\couter$ and $\vouter$ in \cref{thm_bad_graph}, respectively.
Now $ G^\star$ is defined as the $1$-sum of $\Ginner$ with all $G^j$ at the respective nodes $j \in V(\Ginner)$ and $v^j \in V(G^j)$, i.e., $G^\star \coloneqq \Ginner \onesum{2}{v^2} G^2 \onesum{3}{v^3} \dotsb \onesum{n}{v^n} G^n$, where the $\oplus$-operator has to be applied from left to right.
Note that we have
\begin{equation*}
  \Rstab(G^\star) = \Rstab(\Ginner) \onesum{2}{v^2} \Rstab(G^2) \onesum{3}{v^3} \dotsb \onesum{n}{v^n} \Rstab(G^n)
\end{equation*}
by Property~\eqref{eq_property_onesum} and~\cref{thm_summary_conditions}~\ref{thm_summary_conditions_reverse}.

\begin{example}[continues=ex_C5_K3]
  For $\Roddcyclestabfive$, $G^\star$ consists of $\Ginner = C_5$ and $G^j = K_3^j$ for $j \in \{2,3,4,5\}$ as depicted in \cref{fig_C5_K3_Gstar}.
\end{example}

\begin{example}[continues=ex_C3_K4]
  For $\Roddcyclestab$, $G^\star$ consists of $\Ginner = C_3$ and $G^j = K_4^j \onesum{A^j}{A^j} K^j_2$ for $j \in \{2,3\}$ as depicted in \cref{fig_C3_K4_Gstar}.
\end{example}

\begin{figure}[htb]
  \begin{subfigure}[b]{0.49\textwidth}
    \centering
    \begin{tikzpicture}
      \node[node] (1) at (0,-0.2) {$1$};
      \node[node] (2) at (1,-1) {$2$};
      \node[node] (3) at (0.75,-2) {$3$};
      \node[node] (4) at (-0.75,-2) {$4$};
      \node[node] (5) at (-1,-1) {$5$};
      \draw[edge] (1) -- (2) -- (3) -- (4) -- (5) -- (1);

      \node[node] (A2) at (2,0) {$A^2$};
      \node[node] (B2) at (3,0.3) {$B^2$};
      \node[node] (C2) at (2.8,-0.7) {$C^2$};
      \draw[edge] (A2) -- (B2) -- (C2) -- (A2);
      \draw[densely dotted] (2) to node[anchor=south east] {$\oplus$} (A2);

      \node[node] (A3) at (2,-2.5) {$A^3$};
      \node[node] (B3) at (2.6,-1.6) {$B^3$};
      \node[node] (C3) at (3,-2.5) {$C^3$};
      \draw[edge] (A3) -- (B3) -- (C3) -- (A3);
      \draw[densely dotted] (3) to node[anchor=south] {$\oplus$} (A3);

      \node[node] (A4) at (-2,-2.5) {$A^4$};
      \node[node] (B4) at (-3,-2.5) {$B^4$};
      \node[node] (C4) at (-2.6,-1.6) {$C^4$};
      \draw[edge] (A4) -- (B4) -- (C4) -- (A4);
      \draw[densely dotted] (4) to node[anchor=south] {$\oplus$} (A4);

      \node[node] (A5) at (-2,0) {$A^5$};
      \node[node] (B5) at (-2.8,-0.7) {$B^5$};
      \node[node] (C5) at (-3,0.3) {$C^5$};
      \draw[edge] (A5) -- (B5) -- (C5) -- (A5);
      \draw[densely dotted] (5) to node[anchor=south west] {$\oplus$} (A5);
    \end{tikzpicture}
    \caption{Construction of $G^\star$ for \cref{ex_C5_K3}.}
    \label{fig_C5_K3_Gstar}
  \end{subfigure}
  \begin{subfigure}[b]{0.49\textwidth}
    \centering
    \begin{tikzpicture}[scale=0.95]
    \foreach \pos/\name in {{(3.5,2.4)/1}, {(4.25,1.25)/2}, {(2.75,1.25)/3}}
    {
      \node[vertex] (\name) at \pos {$\name$};
    }

    \node[vertex] (v2) at (5.25,1.25) {\scriptsize $A'^2$};

    \draw[densely dotted] (2) to node[anchor=south] {$\oplus$} (v2);


    \foreach \pos/\name in {{(5.25,0.25)/12}, {(6.5,0.25)/22}, {(5.25,-1)/42}, {(6.5,-1)/32}}
    {
      \node[vertex] (\name) at \pos {};
    }

    \node at (5.25,0.25) {\scriptsize $A^2$};
    \node at (6.5,0.25) {\scriptsize $B^2$};
    \node at (6.5,-1) {\scriptsize $C^2$};
    \node at (5.25,-1) {\scriptsize $D^2$};


    \node[vertex] (v3) at (1.75,1.25) {\scriptsize $A'^3$};

    \draw[densely dotted] (3) to node[anchor=south] {$\oplus$} (v3);

    \foreach \pos/\name in {{(1.75,0.25)/13},{(0.5,0.25)/43}, {(1.75,-1)/23},{(0.5,-1)/33}}
    {
      \node[vertex] (\name) at \pos {};
    }

    \node at (1.75,0.25) {\scriptsize $A^3$};
    \node at (1.75,-1) {\scriptsize $B^3$};
    \node at (0.5,-1) {\scriptsize $C^3$};
    \node at (0.5,0.25) {\scriptsize $D^3$};


    \foreach \source/ \dest  in {12/v2,13/v3}
      \path[edge] (\source) -- (\dest);

    \foreach \source/ \dest  in {1/3, 3/2, 2/1}
      \path[edge] (\source) -- (\dest);

    \foreach \source/ \dest  in {12/22, 12/32, 12/42, 22/32, 22/42, 32/42}
      \path[edge] (\source) -- (\dest);

    \foreach \source/ \dest  in {13/43, 13/23, 13/33, 43/23, 43/33, 23/33}
      \path[edge] (\source) -- (\dest);
  \end{tikzpicture}
  \caption{Construction of $G^\star$ for \cref{ex_C3_K4}.}
    \label{fig_C3_K4_Gstar}
  \end{subfigure}
  \caption{%
    Construction of graph $G^\star$ for the two running examples as a $1$-sum of $\Ginner$ and copies of $\Gouter$.
    Pairs of nodes that are identified in the $1$-sums are connected by dotted lines.
  }
\end{figure}
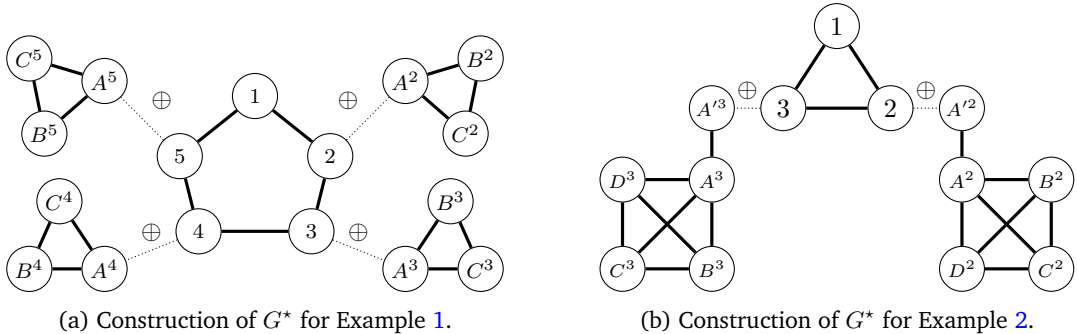

\subsection{The objective vector}

It remains to construct an objective vector $ c^\star \in \R^{V(G^\star)} $ that shows that $ \Rstab(G^\star) $ does not have the persistency property. 
As in the previous section, $c^j \in \R^{V(G^j)}$ and $v^j \in V(G^j)$ denote the vector and node corresponding to $\couter$ and $\vouter$ in \cref{thm_bad_graph}, for each $j \in \{2,3,\dots,n\}$.
Let $ A, b $ be as in the previous section, and denote by $a \coloneqq A_{1,\star}$ the first row of $A$. 
We will define $ c^\star $ via
\[
  c^\star_1 \coloneqq \varepsilon \quad \text{ and } \quad c^\star_v \coloneqq a_j \cdot c^j_v \text{ for all } v \in V(G^j), \, j \in \setdef{2,3,\dotsc,n},
\]
where $ \varepsilon > 0 $ is a positive constant that we will define later. Our first claim is independent of the specific choice of $ \varepsilon $.

\begin{claim}
  \label{thm_integer_optima}
  Every $c^\star$-maximal stable set in $ G^\star $ contains node $1 \in V(\Ginner)$.
\end{claim}

\begin{proof}
  By \cref{thm_bad_graph} there exists, for each $j \in \{2,3,\dotsc,n\}$, a $c^j$-maximal stable set $S^j \subseteq V(G^j)$ that does not use $v^j$.
  Thus, the maximum objective value obtained on $V(G^\star \setminus \{1\})$ is $\sum_{j=2}^n a_j c^j(S^j)$, which is equal to the maximum objective value for all stable sets that do not contain node $1$.
  Since $v^j \notin S^j$ for each $j$, the set $S^\star \coloneqq \bigcup_{j=2}^n S^j \cup \{1\}$ is a stable set in $ G^\star $ with objective value $\varepsilon + \sum_{j=2}^n a_j c^j(S^j) > \sum_{j=2}^n a_j c^j(S^j)$, which proves the claim.
\end{proof}

Again, we verify \cref{thm_integer_optima} for our two running examples.

\begin{example}[continues=ex_C5_K3]
  For $\Roddcyclestabfive$, $Ax \leq b$ consists only of the odd-cycle inequality $x_1 + x_2 + x_3 + x_4 + x_5 \leq 2$, we have $a_j = 1$ for $j \in \{2,3,4,5\}$.
  Hence, the objective vector is defined via $c^\star_1 = \varepsilon$, $c^\star_v = 1$ for all other nodes $v \in V(G^\star) \setminus \{1\}$, where a specific value of $\varepsilon$ still has to be defined.
  Each $c^\star$-maximal stable set in $G^\star$ must contain node $1$ since otherwise it would contain nodes $A^2$ or $A^3$, which we could replace by $B^2$ or $B^3$ without a decrease of the objective value.
  This in turn allows to include node $1$ as well.
\end{example}

\begin{example}[continues=ex_C3_K4]
  For $\Roddcyclestab$, $Ax \leq b$ consists only of the triangle inequality $x_1 + x_2 + x_3 \leq 1$, we have $a_j = 1$ for $j \in \{2,3\}$.
  Hence, the objective vector is defined via $c^\star_1 = \varepsilon$, $c^\star_j = \frac{1}{3}$, $c^\star_{A^j} = 1+\frac{2}{3}$ and $c^\star_{B^j} = c^\star_{C^j} = c^\star_{D^j} = 1$ for $j \in \{2,3\}$, where a specific value of $\varepsilon$ still has to be defined.
  It is easy to see that $\{1, A^2, A^3\}$ is the unique $c^\star$-maximal stable set in $G^\star$.
\end{example}

To see that $ \Rstab(G^\star) $ does not have the persistency property, it suffices to establish the following claim, which then yields  \cref{thm_main}.

\begin{claim}
  \label{thm_optimal_lp_point}
  For $ \varepsilon > 0 $ small enough, every $ c^\star $-optimal point $ x^\star \in \Rstab(G^\star) $ satisfies $ x^\star_1 = 0 $.
\end{claim}

Let $ x^\star $ be any $ c^\star $-optimal point in $ \Rstab(G^\star) $.
Observe that by construction, for each $j \in \{2,\dotsc,n\} $, $v^j$ is the only node connecting $G^j$ to the rest of nodes $v \in V(G^\star) \setminus V(G^j)$. 
Hence, for every fixed value of $ x^\star_{v^j} $, the maximal value of $\sum_{v \in V(G^j)} c^\star_v x_v$ 
can be computed independently of the values of the variables corresponding to nodes $v \in V(G^\star) \backslash V(G^j)$.
In order to understand the contributions of the variables corresponding to nodes $v \in V(G^j)$ to the total optimal value in terms of $ x^\star_{v^j} $, let us introduce the function $f : [0,1] \to \R$ defined via
\[
  f(y) \coloneqq \max \bigsetdef{ \transpose{{c^j}}x }[ x \in \Rstab(G^j) \text{ and } x_{v^j} = y ].
\]

Note that the definition is independent of $j$ since all $(G^j,c^j,v^j)$ are identical up to indexing. We observe that the restriction of $ x^\star $ onto the coordinates corresponding to $ V(\Ginner) $ is an optimal solution for
\begin{equation}
  \max \bigsetdef{ c'(x) }[ x \in \Rstab(\Ginner) ]
  = \max \bigsetdef{ c'(x) }[ x \in \Redgestab(\GinnerEdge), \, Ax \leq b ], \label{eq_alternative_opt}
\end{equation}
where $c'(x) \coloneqq \varepsilon x_1 + \sum_{j=2}^n a_j f(x_j)$.
Thus, we see that \cref{thm_optimal_lp_point} immediately follows from the following result.

\begin{claim}
\label{thm_optimal_lp_point_ginner}
  For $ \varepsilon > 0 $ small enough, every $ c' $-optimal point $ x \in \Rstab(\Ginner) $ satisfies $ x_1 = 0 $.
\end{claim}

\begin{example}[continues=ex_C5_K3]
  For $\Roddcyclestabfive$, we have
  \begin{equation*}
    f(x_j) = \max \bigsetdef{ x_{A^j} + x_{B^j} + x_{C^j}}[ x \in \Roddcyclestabfive(K_3^j) \text{ and } x_{A^j} = x_j ]
  \end{equation*}
  for $j=2,3,4,5$, and due to $\Roddcyclestabfive(K_3^j) = \Redgestab(K_3^j)$ it follows that $f$ attains its unique maximum at $x_j = \frac{1}{2}$.
  Consequently, any vector $x^{\text{LP}} \in \Rstab(G^\star)$ with $x^{\text{LP}}_v = \frac{1}{2}$ for all $v \in V(G^\star) \setminus \{1\}$ is $c^\star$-maximal if we ignore the objective contribution of $\varepsilon \cdot x_1$ (see \cref{fig_C5_K3_data} for an illustration).
  Now, setting $x_2 = x_3 = x_4 = x_5 = \frac{1}{2}$ leaves no slack in the odd-cycle inequality $x_1 + x_2 + x_3 + x_4 + x_5 \leq 2$.
  Hence, a positive $x_1$-variable would require a reduction of $x_j$ for some $j \in \{2,3,4,5\}$, which in turn reduced $f(x_j)$.
  Hence, for sufficiently small $\varepsilon > 0$, such a reduction is not profitable, which proves \cref{thm_optimal_lp_point_ginner} for this example.
  \cref{fig_C5_K3_data} depicts $G^\star$, $c^\star$, a $c^\star$-optimal point $x^\star \in \Rstab(G^\star)$ and a $c^\star$-maximal stable set for $\varepsilon = \frac{1}{20}$.
\end{example}

Actually, the fact that there is no slack in the odd cycle inequality to set $x_1 > 0$ in \cref{ex_C5_K3} is not a coincidence, it follows from the following result by Sewell on the defect of facets of the stable set polytope.

\begin{proposition}[Corollary~3.4.3 in~\cite{Sewell90}]
  \label{thm_defect}
  Let $\sum_{j=1}^n a_j x_j \leq b_1$ be a facet-defining inequality for the stable set polytope of a graph on $n$ nodes that is neither a bound nor an edge inequality.
  Then we have
  \[
    a_1 \leq \sum_{j=1}^n a_j - 2b_1.
  \]
\end{proposition}

\begin{example}[continues=ex_C3_K4]
  For $\Roddcyclestab$, we have
  \[
    f(x_j) = \max \bigsetdef{ \frac13 x_{A'^j} + \frac{5}{3} x_{A^j} + x_{B^j} + x_{C^j} + x_{D^j}}[ x \in \Roddcyclestab(G^j) \text{ and } x_{A'^j} = x_j ],
  \]
  for $j = 2,3$, where $G^j = K_4^j \onesum{A^j}{A^j} K_2^j$, and for each $j$, $f$ attains its unique maximum at $\hat{x}_v = \frac{1}{3}$ for $v \in \{A^j,B^j,C^j,D^j\}$ and $\hat{x}_{A'^j} = \frac{2}{3}$.
  However, in this case, setting $x_j = \frac23$ for $j=2,3$ results in an infeasible solution of optimization problem~(\ref{eq_alternative_opt}), hence \cref{thm_optimal_lp_point_ginner} does not follow as easily as for \cref{ex_C5_K3}.
\end{example}

As illustrated by \cref{ex_C3_K4}, the general proof of \cref{thm_optimal_lp_point_ginner} is a bit more technical than for \cref{ex_C5_K3} since we have to ensure that all inequalities $Ax \leq b$ and all edge inequalities of $\GinnerEdge$ are satisfied, which is not always the case for the optimal solutions obtained when considering $f(x_j)$ separately for each $j$.
To overcome this difficulty for the first inequality $\transpose{a}x \leq b_1$ of the system $Ax \leq b$, it will be convenient to consider the function $g : [0,\infty] \to \R$ defined via
\[
  g(z) \coloneqq \max \bigsetdef{ \sum_{j=2}^n a_j f(x_j) }[ \transpose{a} x \leq z,~ x \in \Redgestab(\GinnerEdge) ].
\]

The intuition behind the proof of \cref{thm_optimal_lp_point_ginner} is the following: First, note that $c'(x)$ is the sum of $\varepsilon x_1$ and the objective function defining $g$.
Function $g(z)$ represents the contribution to the objective value of $G^j$ for $j=2,\dots, n$ as a function of the right-hand side of the inequality $\transpose{a} x \leq z$.
Using \cref{thm_defect} we will soon prove that $g(z)$ is strictly increasing on the interval $z \in [0,b_1]$.
Since $a_1 > 0$ and $x_1$ does not contribute to the maximum in the definition of $g$, the latter is attained only by solutions $x$ with $x_1 = 0$.
If we ignore, for a moment, the inequalities $Ax \leq b$, this shows that for sufficiently small $\varepsilon$, also every $c'$-maximal solution $x$ satisfies $x_1 = 0$.
The formal steps are as follows.

\begin{claim}
  \label{thm_f_g_concave}
  The functions $ f $ and $ g $ are concave. Moreover, $ g $ is strictly monotonically increasing on $ [0, b_1] $.
\end{claim}

\begin{example}[continues=ex_C3_K4]
  For $\Roddcyclestab$, we have
  \[
    g(z) \coloneqq \max \bigsetdef{ \sum_{j \in \{2,3\}} f(x_j) }[ x_1 + x_2 + x_3 \leq z,~ x \in \Redgestab(C_3) ].
  \]
  Function $g$ is illustrated in Figure~\ref{fig_g_ex_C3_K4}. It is clearly concave, and linear and strictly monotonically increasing on $[0,b_1] = [0,1]$, hence \cref{thm_f_g_concave} is satisfied.
\end{example}
%
%
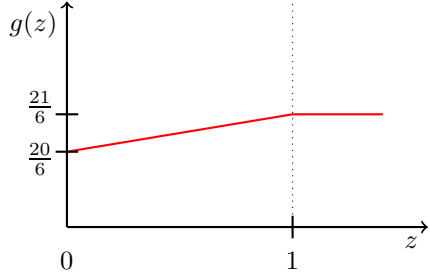
\begin{figure}[htb]
  \centering
  \begin{tikzpicture}[scale=3]
    \draw[thick,->] (0,3) -- (1.6,3) node[anchor=north east] {$z$};
    \draw[thick,->] (0,3) -- (0,4) node[anchor=north east] {$g(z)$};

    \node at (0,2.85) {$0$};

    \draw[thick] (1,2.95) -- (1,3.05);
    \node at (1,2.85) {$1$};

    \draw[red, thick] (0,3.334) -- (0.2,3.3672) -- (0.4,3.4004) -- (0.6,3.4336) -- (0.8,3.4668) -- (1,3.5) -- (1.2,3.5) -- (1.4,3.5);

    \draw[dotted] (1,3) -- (1,4);

    \draw[thick] (-0.05,3.5) -- (0.05,3.5);
    \node at (-0.12,3.533) {$\frac{21}{6}$};

    \draw[thick] (-0.05,3.334) -- (0.05,3.334);
    \node at (-0.12,3.3) {$\frac{20}{6}$};
  \end{tikzpicture}
  \caption{Illustration of function $g(z)$ for Example~\ref{ex_C3_K4}.}
  \label{fig_g_ex_C3_K4}
\end{figure}

\begin{proof}[Proof of \cref{thm_optimal_lp_point_ginner}]
Letting
\begin{align*}
  \gamma &\coloneqq \min \bigsetdef{ x_1 }[ x \text{ vertex of } \Rstab(G^\star) \text{ with } x_1 > 0 ] \in (0,1], \text{ and} \\
  \lambda &\coloneqq \min \bigsetdef{ \gamma / (A_{i,1} + \dots + A_{i,n}) }[ i \in \{1,2,\dotsc,m\} ] \in (0,1),
\end{align*}
we claim that every choice of $ \varepsilon $ with
\[
  0 < \varepsilon < \lambda (g(b_1) - g(b_1 - a_1 \gamma))
\]
satisfies the assertion.
First, we need to verify that the right-hand side of the inequality above is positive.
To this end, note that $ a_1 \le b_1 $ and hence $ 0 \leq b_1 - a_1 \gamma < b_1 $.
So, by \cref{thm_f_g_concave} we have
\begin{equation}
  \label{ineq_g}
  g(b_1 - a_1 \gamma) < g(b_1),
\end{equation}
which yields positivity of the right-hand side.

Next, let $ \varepsilon $ be as above.
For the sake of contradiction, assume that there exists a $ c' $-optimal solution $ x^\star \in \Rstab(\Ginner) $ with $ x^\star_1 > 0 $.
Note that $ x^\star $ can be extended to a $ c^\star $-optimal solution over $ \Rstab(G^\star) $, which we may assume to be a vertex of $ \Rstab(G^\star) $, and hence $ x^\star_1 \ge \gamma $.
Let $\hat{x}^0 \in \Rstab(\Ginner) $ be equal to $x^\star$, except for $\hat{x}^0_1 \coloneqq 0$.
Moreover, let $\hat{x}^1 \in \R^{V(\Ginner)}$ be a maximizer of~$ g(b_1) $, which may not be contained in $ \Rstab(\Ginner) $.
Now consider the vector $\hat{x}^\lambda \coloneqq (1-\lambda) \hat{x}^0 + \lambda \hat{x}^1$.
To obtain the desired contradiction, we will show that $ \hat{x}^\lambda $ is contained in $ \Rstab(\Ginner) $ and that $ c'(\hat{x}^\lambda) > c'(x^\star) $.

Since $\hat{x}^0$ and $\hat{x}^1$ both lie in $\Redgestab(\GinnerEdge)$, also $x^\lambda$ lies in $\Redgestab(\GinnerEdge)$.
Let $i \in \setdef{1,2,\dotsc,m}$.
By \cref{thm_full_support}, $A_{i,1} \geq 1$ holds, which implies $A_{i,\star} \hat{x}^0 \leq A_{i,\star} x^\star - \gamma \leq b_i - \gamma$.
We obtain
\[
  A_{i,\star} \hat{x}^\lambda
  = A_{i,\star} \hat{x}^0 + \lambda A_{i,\star} (\hat{x}^1 - \hat{x}^0)
  \leq b_i - \gamma + \lambda (A_{i,1} + \dots + A_{i,n})
  \leq b_i,
\]
where the second inequality follows from the fact that each coordinate of $ \hat{x}^1 - \hat{x}^0 $ is bounded by $ 1 $, and the last inequality holds by the definition of $ \lambda $.
This shows that $\hat{x}^\lambda$ is contained in $\Rstab(\Ginner)$.

For the objective value of $\hat{x}^1$ we clearly have $c'(\hat{x}^1) \ge g(b_1)$.
Moreover, since $ \hat{x}^0_1 = 0 $ we have
\[
  c'(\hat{x}^0) \le g(\transpose{a} \hat{x}^0) \le g(b_1 - a_1 \gamma) < g(b_1),
\]
where the latter two inequalities again follow from \cref{thm_f_g_concave} and~\eqref{ineq_g}.
Observe that concavity of $f$ and nonnegativity of $a$ imply concavity of $c'(x)$, which yields
$c'(\hat{x}^\lambda) \geq (1-\lambda) c'(\hat{x}^0) + \lambda c'(\hat{x}^1)$.
We obtain
\begin{align*}
  c'(x^\star) - c'(\hat{x}^\lambda)
  &\leq \big( \varepsilon + c'(\hat{x}^0) \big) - \big( c'(\hat{x}^0) - \lambda (c'(\hat{x}^0) - c'(\hat{x}^1)) \big)
  = \varepsilon + \lambda (c'(\hat{x}^0) - c'(\hat{x}^1)) \\
  &\leq \varepsilon + \lambda ( g(b_1 - a_1 \gamma) - g(b_1) )
  < 0,
\end{align*}
where the last inequality holds by definition of $\varepsilon$ and due to~\eqref{ineq_g}.
\end{proof}

\begin{example}[continues=ex_C3_K4]
  For $\Roddcyclestab$, one can check (for example with a computer program) that $\gamma = \frac{1}{3}$ and $\lambda = \frac{1}{9}$.
  For every $0 < \varepsilon < \frac{1}{9} \left( g(1) - g(1 - \frac{1}{3}) \right) = \frac{1}{9} \left( \frac{7}{2} - \frac{31}{9} \right) = \frac{1}{162}$, \cref{thm_optimal_lp_point_ginner} is satisfied.
  \cref{fig_C3_K4_data} depicts $G^\star$, $c^\star$, a $c^\star$-optimal point $x^\star \in \Rstab(G^\star)$ and a $c^\star$-maximal stable set for $\varepsilon = \frac{1}{300}$.
\end{example}

\begin{figure}[htb]
  \begin{subfigure}[b]{0.53\textwidth}
    \centering
    \begin{tikzpicture}
      \node[node,very thick,fill=black!10!white] (1) at (0,0.2) {$1$};
      \node[node] (2) at (1.5,-1) {$2$};
      \node[node] (3) at (0.75,-2.5) {$3$};
      \node[node] (4) at (-0.75,-2.5) {$4$};
      \node[node] (5) at (-1.5,-1) {$5$};
      \draw[edge] (1) -- (2) -- (3) -- (4) -- (5) -- (1);

      \node[above=0.05cm of 1] {\small $\frac{1}{20}$;$0$;$1$};
      \node[below=0.15cm of 2, anchor=west] {\small $1$;$\frac12$;$0$};
      \node[below=0.05cm of 3] {\small $1$;$\frac12$;$0$};
      \node[below=0.05cm of 4] {\small $1$;$\frac12$;$0$};
      \node[below=0.15cm of 5, anchor=east] {\small $1$;$\frac12$;$0$};

      \node[node,very thick,fill=black!10!white] (B2) at (2.125,0.1) {$B^2$};
      \node[node] (C2) at (2.75,-1.0) {$C^2$};
      \draw[edge] (2) -- (B2) -- (C2) -- (2);
      \node[right=0.05cm of B2] {\small $1$;$\frac12$;$1$};
      \node[right=0.05cm of C2] {\small $1$;$\frac12$;$0$};

      \node[node,very thick,fill=black!10!white] (B3) at (2.125,-2.5) {$B^3$};
      \node[node] (C3) at (1.5,-3.6) {$C^3$};
      \draw[edge] (3) -- (B3) -- (C3) -- (3);
      \node[right=0.05cm of B3] {\small $1$;$\frac12$;$1$};
      \node[right=0.05cm of C3] {\small $1$;$\frac12$;$0$};

      \node[node,very thick,fill=black!10!white] (B4) at (-2.125,-2.5) {$B^4$};
      \node[node] (C4) at (-1.5,-3.6) {$C^4$};
      \draw[edge] (4) -- (B4) -- (C4) -- (4);
      \node[left=0.05cm of B4] {\small $1$;$\frac12$;$1$};
      \node[left=0.05cm of C4] {\small $1$;$\frac12$;$0$};

      \node[node,very thick,fill=black!10!white] (B5) at (-2.125,0.1) {$B^5$};
      \node[node] (C5) at (-2.75,-1.0) {$C^5$};
      \draw[edge] (5) -- (B5) -- (C5) -- (5);
      \node[left=0.05cm of B5] {\small $1$;$\frac12$;$1$};
      \node[left=0.05cm of C5] {\small $1$;$\frac12$;$0$};
    \end{tikzpicture}
    \caption{$G^\star$, $c^\star$ and LP/IP maxima for \cref{ex_C5_K3}.}
    \label{fig_C5_K3_data}
  \end{subfigure}
  \begin{subfigure}[b]{0.43\textwidth}
  \centering
  \begin{tikzpicture}[auto,swap]
    \foreach \pos/\name in {{(3.5,2.65)/1}, {(4.25,1.5)/2}, {(2.75,1.5)/3}}
    {
      \node[vertex] (\name) at \pos {$\name$};
    }

    \node[node,very thick,fill=black!10!white] at (3.5,2.65) {$1$};

    \node[above=0.05cm of 1] {$\frac{1}{300}$;$0$;$1$};
    \node[right=0.05cm of 2] {$\frac{1}{3}$;$\frac{2}{3}$;$0$};
    \node[left=0.05cm of 3] {$\frac{1}{3}$;$\frac{1}{3}$;$0$};

    \foreach \pos/\name in {{(4.25,0.25)/12}, {(5.5,0.25)/22},
                            {(4.25,-1)/42}, {(5.5,-1)/32}}
    {
      \node[vertex] (\name) at \pos {};
    }

    \node[node,very thick,fill=black!10!white] at (4.25,0.25) {\scriptsize $A^2$};
    \node at (5.5,0.25) {\scriptsize $B^2$};
    \node at (5.5,-1) {\scriptsize $C^2$};
    \node at (4.25,-1) {\scriptsize $D^2$};

    \node[above right=0.005cm of 12] {\small $\frac{5}{3}$;$\frac13$;$1$};
    \node[right=0.05cm of 22] {\small $1$;$\frac13$;$0$};
    \node[below=0.05cm of 32] {\small $1$;$\frac13$;$0$};
    \node[below=0.05cm of 42] {\small $1$;$\frac13$;$0$};

    \foreach \pos/\name in {{(2.75,0.25)/13},{(1.5,0.25)/43},
                            {(2.75,-1)/23},{(1.5,-1)/33}}
    {
      \node[vertex] (\name) at \pos {};
    }

    \node[node,very thick,fill=black!10!white] at (2.75,0.25) {\scriptsize $A^3$};
    \node at (2.75,-1) {\scriptsize $B^3$};
    \node at (1.5,-1) {\scriptsize $C^3$};
    \node at (1.5,0.25) {\scriptsize $D^3$};

    \node[above left=0.005cm of 13] {\small $\frac{5}{3}$;$\frac23$;$1$};
    \node[below=0.05cm of 23] {\small $1$;$\frac16$;$0$};
    \node[below=0.05cm of 33] {\small $1$;$\frac16$;$0$};
    \node[left=0.05cm of 43] {\small $1$;$\frac16$;$0$};

    \foreach \source/ \dest  in {12/2,13/3}
      \path[edge] (\source) -- (\dest);

    \foreach \source/ \dest  in {1/3, 3/2, 2/1}
      \path[edge] (\source) -- (\dest);

    \foreach \source/ \dest  in {12/22, 12/32, 12/42, 22/32, 22/42, 32/42}
      \path[edge] (\source) -- (\dest);

    \foreach \source/ \dest  in {13/43, 13/23, 13/33, 43/23, 43/33, 23/33}
      \path[edge] (\source) -- (\dest);
  \end{tikzpicture}
  \caption{$G^\star$, $c^\star$ and LP/IP maxima for \cref{ex_C3_K4}.}
  \label{fig_C3_K4_data}
  \end{subfigure}
  \caption{%
    Illustration of $G^\star$ for our running examples.
    For each node $v \in V(G^\star)$, the triple $c_v^\star;x_v^{\text{LP}};x_v^{\text{IP}}$ denotes the corresponding coefficient in $c^\star$, the value of component $v$ in a $c^\star$-maximal solution over $\Rstab(G^\star)$ and the value in the unique $c^\star$-maximal stable set, respectively.
  }
  \label{fig_data}
\end{figure}
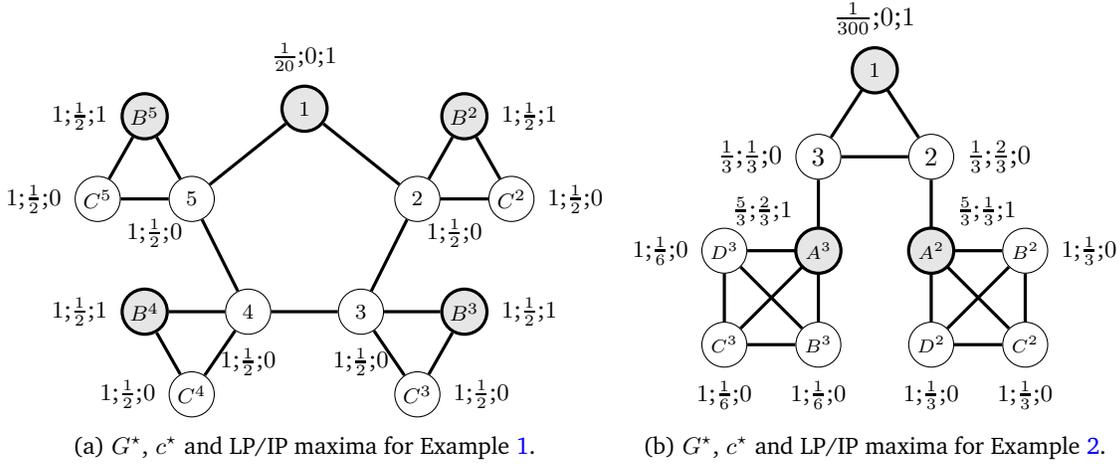

To conclude the proof of \cref{thm_main}, it remains to prove \cref{thm_f_g_concave}.
The fact that $ f $ and $ g $ are concave is a simple consequence of the next basic lemma.
\begin{restatable}{lemma}{rspolytopevaluefunction}
  \label{thm_polytope_value_function}
  Let $P \subseteq \R^n$ be a non-empty polytope, let $c,a \in \R^n$ and let $\ell \coloneqq \min \setdef{ \transpose{a} x }[ x \in P ]$.
  The functions $h^=, h^{\leq} : [\ell,\infty) \to \R$ defined via $h^=(\beta) = \max \setdef{ \transpose{c}x }[ x \in P, \transpose{a}x = \beta ]$ and $h^{\leq}(\beta) = \max \setdef{ \transpose{c}x }[ x \in P, \transpose{a}x \le \beta ]$ are concave.
  Moreover, there exists a number $\beta^\star \in [\ell,\infty)$ such that $h^=$ and $h^\le$ are identical and strictly monotonically increasing on the interval $[\ell, \beta^\star]$, and $ h^\le $ is constant on the interval $[\beta^\star, \infty)$.
\end{restatable}
\begin{proof}
  See \cref{appendix_proofs}.
\end{proof}

\begin{proof}[Proof of \cref{thm_f_g_concave}]
From \cref{thm_polytope_value_function} it is clear that $ f $ is concave.
By rewriting
\[
  g(z) = \max \big\{ \sum_{j=2}^n a_j \cdot \hspace{-0.5em} \sum_{v \in V(G^j)} \hspace{-0.5em} c^j_v x_v \mid \sum_{j=1}^n a_j x_j \leq z,~ x \in \Redgestab(\GinnerEdge) \onesum{2}{v^2} \Rstab(G^2) \onesum{3}{v^3} \dotsb \onesum{n}{v^n} \Rstab(G^n) \big\},
\]
we also see that $ g $ is concave.
Moreover, again by \cref{thm_polytope_value_function}, there exists some $ \beta^\star \ge 0 $ such that $ g $ is strictly monotonically increasing on the interval $ [0,\beta^\star] $, and constant on $ [\beta^\star, \infty) $.
It suffices to show that $ \beta^\star \ge b_1 $.
To this end, let us get back to our initial definition of $ g $, and let $ \hat{x} \in \Redgestab(\GinnerEdge) $ be a maximizer for $ g(\infty) $.
Note that $ \beta^\star \geq \transpose{a} \hat{x} $ by definition of $\beta^\star$, and hence we have to show that $\hat{x}$ satisfies $\transpose{a} \hat{x} \geq b_1$.

Since the objective value of $\hat{x}$ does not depend on $\hat{x}_1$, we may assume that $\hat{x}_1 = 0$.
By the construction of $ G^j $ and $ c^j $, we know that $ f $ attains its unique maximum at $y^\star \geq \frac{1}{2}$.
This implies $ 0 \leq \hat{x}_j \leq y^\star$ for $j=2,3,\dotsc,n$.
Moreover, we claim that also $\hat{x}_j \geq 1 - y^\star$ holds.
Suppose not, then none of the edge inequalities involving $x_j$ is tight.
Then $\hat{x}_j < 1-y^\star \leq y^\star$ shows that increasing $\hat{x}_j$ would improve the objective value, which in turn contradicts optimality of $\hat{x}$.
Consequently, even $1-y^\star \leq \hat{x}_j \leq y^\star$ holds for $j=2,3,\dotsc,n$.

Let $J(\alpha) \coloneqq \setdef{ 2 \leq j \leq n }[ \hat{x}_j = \alpha ]$ for $\alpha \in [1-y^\star,y^\star]$.
We will show that $a(J(\alpha)) \geq a(J(1-\alpha))$ holds for all $\alpha \in (1/2, y^\star]$, where $a(J(\alpha))$ shall denote $\sum_{j \in J(\alpha)} a_j$.
Note that this implies the claim since for each $ \alpha \in (1/2, y^\star] $ we then have

\begin{align*}
  \sum_{j \in J(\alpha)} a_j \hat{x}_j + \sum_{j \in J(1-\alpha)} a_j \hat{x}_j
  & = \sum_{j \in J(\alpha)} a_j \alpha + \sum_{j \in J(1-\alpha)} a_j (1 - \alpha) \\
  & = \alpha \cdot [ \underbrace{a(J(\alpha)) - a(J(1-\alpha))}_{\ge 0}] + a(J(1-\alpha)) \\
  & \ge \tfrac{1}{2} \cdot [ a(J(\alpha)) - a(J(1-\alpha))] + a(J(1-\alpha)) \\
  &= \sum_{j \in J(\alpha)} a_j \tfrac{1}{2} + \sum_{j \in J(1-\alpha)} a_j \tfrac{1}{2}
\end{align*}
and hence
\begin{align*}
  \transpose{a} \hat{x}
  =\sum_{j=2}^n a_j \hat{x}_j
  & = \sum_{j \in J(1/2)} a_j \hat{x}_j + \sum_{\alpha \in (1/2,y^\star]} \big( \sum_{j \in J(\alpha)} a_j \hat{x}_j + \sum_{j \in J(1-\alpha)} a_j \hat{x}_j \big) \\
  &\ge \sum_{j \in J(1/2)} a_j \tfrac{1}{2} + \sum_{\alpha \in (1/2,y^\star]} \big( \sum_{j \in J(\alpha)} a_j \tfrac{1}{2} + \sum_{j \in J(1-\alpha)} a_j \tfrac{1}{2} \big)
  = \tfrac{1}{2} \sum_{j=2}^n a_j
  \geq b_1,
\end{align*}
where the last inequality follows from \cref{thm_defect}.

For the sake of contradiction, assume that $a(J(\alpha)) < a(J(1-\alpha))$ holds for some $\alpha \in (1/2,y^\star]$.
For a sufficiently small $\varepsilon' > 0$, the solution $\hat{x}' \in \R^{V(\Ginner)}$ defined via
\[
  \hat{x}'_j \coloneqq \begin{cases}
    \hat{x}_j + \varepsilon' &\text{ if } j \in J(1-\alpha) \\
    \hat{x}_j - \varepsilon' &\text{ if } j \in J(\alpha) \\
    \hat{x}_j \text{ otherwise}
  \end{cases} \text{ for } j=1,2,\dotsc,n
\]
is still contained in $ \Redgestab(\GinnerEdge) $.
To see this, observe that $\hat{x}'_j \geq 0$ holds for all $j \in V(\Ginner)$ since we only decrease entries that are at least $1/2$.
Moreover, edge inequalities for $\GinnerEdge$ that are tight for $\hat{x}$ remain tight for $\hat{x}'$, since either none or both of its two node values are modified, where in the latter case, the value is increased by $\varepsilon'$ for one node and decreased by $\varepsilon'$ for the other.
Finally, edge inequalities that are not tight for $\hat{x}$ will not be violated if we choose $\varepsilon'$ sufficiently small.
For the objective values we obtain
\begin{align*}
  \sum_{j=2}^n a_j (f(\hat{x}'_j) - f(\hat{x}_j))
  &= \sum_{j \in J(1-\alpha)} a_j (f(\hat{x}'_j) - f(\hat{x}_j)) + \sum_{j \in J(\alpha)} a_j (f(\hat{x}'_j) - f(\hat{x}_j)) \\
  &= a(J(1-\alpha)) \cdot \big( f(1-\alpha + \varepsilon') - f(1-\alpha) \big) + a(J(\alpha)) \cdot \big( f(\alpha - \varepsilon') - f(\alpha) \big).
\end{align*}
We also assume that $\varepsilon'$ is small enough to guarantee $1-\alpha+\varepsilon' < \alpha - \varepsilon'$.
Since $f$ is concave and monotonically increasing in $[0,y^\star]$, we obtain $f(1-\alpha+\varepsilon') - f(1-\alpha) \geq f(\alpha) - f(\alpha - \varepsilon')$.
Together with the assumption $a(J(1-\alpha)) > a(J(\alpha))$, this shows that the objective value of $\hat{x}'$ is strictly larger than that of $\hat{x}$, a contradiction to the optimality of $\hat{x}$ (see \cref{fig_concave_charging} for an illustration).
\end{proof}

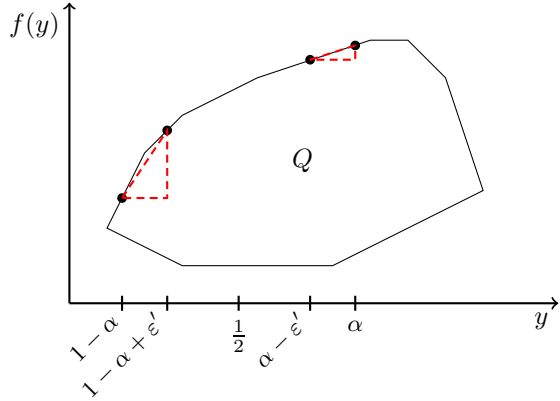
\begin{figure}[htb]
  \centering
  \begin{tikzpicture}
    \draw[thick,->] (-1.5,-0.5) -- (5,-0.5) node[anchor=north east] {$y$};
    \draw[thick,->] (-1.5,-0.5) -- (-1.5,3.5) node[anchor=north east] {$f(y)$};

    \node at (1.6,1.4) {$Q$};

    \draw[thick] (-0.8,-0.4) -- +(0.0,-0.2) node[rotate=45,anchor=east] {\small{$1-\alpha$}};
    \draw[thick] (-0.2,-0.4) -- +(0.0,-0.2) node[rotate=45,anchor=east] {\small{$1-\alpha+\varepsilon'$}};
    \draw[thick] (1.7,-0.4) -- +(0.0,-0.2) node[rotate=45,anchor=east] {\small{$\alpha-\varepsilon'$}};
    \draw[thick] (2.3,-0.4) -- +(0.0,-0.2) node[anchor=north] {\small{$\alpha$}};
    \draw[thick] (0.75,-0.4) -- +(0.0,-0.2) node[anchor=north] {$\tfrac{1}{2}$};

    \draw (0,0) -- (-1,0.5) -- (-0.5,1.5) -- (0,2) -- (1,2.5) -- (2.5,3) -- (3,3) -- (3.5,2.5) -- (4,1) -- (2,0) -- cycle;

    \node[draw=black,fill=black,thick,circle,inner sep=0mm,minimum size=1mm] at (-0.8,0.9) {};
    \node[draw=black,fill=black,thick,circle,inner sep=0mm,minimum size=1mm] at (-0.2,1.8) {};
    \node[draw=black,fill=black,thick,circle,inner sep=0mm,minimum size=1mm] at (1.7,2.74) {};
    \node[draw=black,fill=black,thick,circle,inner sep=0mm,minimum size=1mm] at (2.3,2.93) {};

    \draw[red,thick,densely dashed] (-0.8,0.9) -- (-0.2,1.8) -- (-0.2,0.9) -- cycle;
    \draw[red,thick,densely dashed] (1.7,2.74) -- (2.3,2.93) -- (2.3,2.74) -- cycle;
  \end{tikzpicture}
  \caption{Illustration of modifications in the proof of \cref{thm_f_g_concave}.}
  \label{fig_concave_charging}
\end{figure}

\section{Necessity of Conditions}
\label{sec_necessity_conditions}

In this section we present examples of stable set formulations that have the persistency property and satisfy two of the three conditions~\eqref{eq_property_facet}--\eqref{eq_property_onesum}.
This shows that our result is strongest possible in the sense that the conditions are actually necessary.

\begin{proposition}
  The relaxation defined via 
  \begin{equation*}
    R^{\text{(A)}}(G) \coloneqq \bigsetdef{ x \in \Redgestab(G) }[ x(C) \leq 1.4 \text{ for each $3$-clique $C$ of $G$}  ]
  \end{equation*}
  satisfies Conditions~\eqref{eq_property_valid} and~\eqref{eq_property_onesum}, violates Condition~\eqref{eq_property_facet}, and has the persistency property.
\end{proposition}

\begin{proof}
  It is easy to see that $R^{\text{(A)}}$ satisfies Conditions~\eqref{eq_property_valid} and~\eqref{eq_property_onesum}.
  It violates Property~\eqref{eq_property_facet} since inequality $x(V(C)) \leq 1.4$ is not facet-defining for $\Pstab(C)$, where $C$ is a $3$-clique.
  It remains to establish the persistency property for $R^{\text{(A)}}$ for which we adapt a persistency proof for $\Redgestab$.

  Let $G \in \graphs$, let $c \in \R^{V(G)}$ be an objective vector and let $\hat{x} \in R^{\text{(A)}}$ be a $c$-maximal point in the relaxation.
  Let $U_0 \coloneqq \{ v \in V(G) \mid \hat{x}_v = 0 \}$, $U_1 \coloneqq \{ v \in V(G) \mid \hat{x}_v = 1 \}$ and $W \coloneqq V(G) \setminus (U_0 \cup U_1)$.
  We have to show that there exists a $c$-maximal stable set in $G$ containing the nodes $U_1$ and not containing nodes $U_0$.

  First observe that each node $u \in U_1$ only has neighbors in $U_0$ since otherwise an edge inequality would be violated.
  This implies that each clique constraint with support in $U_0 \cup U_1$ and in $W$ has no support in $U_1$ and we claim that no such constraint is tight.
  For a $2$-clique $\{u,w\} \in E(G)$ with $u \in U_0$ and $w \in W$, we have $\hat{x}_u = 0$ and $0 < \hat{x}_w < 1$.
  Hence, the edge inequality $x_u + x_w \leq 1$ is strictly satisfied.
  For a $3$-clique $C$ with nodes in $U_0$ and in $W$ we have $\hat{x}(V(C)) < 1.4$ since at least one node belongs to $U_0$ and the sum of values of the nodes in $W$ can be at most $1$.
  For any LP, the removal of constraints that are not tight does not change optimality.
  We conclude that $\hat{x}|_{U_0 \cup U_1}$ is an optimum LP solution on the subgraph $G[U_0 \cup U_1]$.

  Let $S$ be a $c$-maximal stable set of $G$ and consider the set $S' \coloneqq (S \setminus U_0) \cup U_1$.
  Since the neighborhood of $U_1$ is contained in $U_0$, $S'$ is a stable set.
  Moreover, by the local optimality of $\hat{x}$, the $c$-value of $S'$ on $U_0 \cup U_1$ is not smaller than that of $S$.
  This shows that also $S'$ is $c$-maximal.
  Hence, $R^{\text{(A)}}$ has the persistency property.
\end{proof}

\begin{proposition}
  The relaxation defined via 
  \begin{equation*}
    R^{\text{(B)}}(G) \coloneqq \begin{cases}
                            \Redgestab(G) & \text{if } G = K_3, \\
                            \Pstab(G) & \text{otherwise}
                         \end{cases} \qquad \text{ for each } G \in \graphs
  \end{equation*}
  satisfies Conditions~\eqref{eq_property_facet} and~\eqref{eq_property_onesum}, violates Condition~\eqref{eq_property_valid}, and has the persistency property.
\end{proposition}

\begin{proof}
  It is easy to see that $R^{\text{(B)}}$ satisfies Conditions~\eqref{eq_property_facet}.
  It violates Property~\eqref{eq_property_valid} since for the graph $G = K_2 \onesum{}{} K_3$, where the $1$-sum is taken at an arbitrary vertex of $K_2$ and an arbitrary vertex of $K_3$,  the clique inequality for the $3$-clique is facet-defining for $\Pstab(G)$, whereas the same inequality is not valid for $R^{\text{(B)}}(K_3)$.
  Since $K_3$ is no proper $1$-sum (i.e., a $1$-sum of two graphs each having at least 2 nodes), we have that for every $G \in \graphs$ that is the proper $1$-sum of two graphs, $R^{\text{(B)}}(G) = \Pstab(G)$ is satisfied, which establishes Property~\eqref{eq_property_onesum}.
  $R^{\text{(B)}}(K_3)$ has the persistency property by \cref{thm_persistency}, and for all other graphs $G$, $R^{\text{(B)}}(G)$ has the persistency property by definition.
\end{proof}

\begin{proposition}
  The relaxation defined via 
  \begin{equation*}
    R^{\text{(C)}}(G) \coloneqq \begin{cases}
                            \Pstab(G) & \text{if } G = K_3, \\
                            \Redgestab(G) & \text{otherwise}
                         \end{cases} \qquad \text{ for each } G \in \graphs
  \end{equation*}
  satisfies Conditions~\eqref{eq_property_facet} and~\eqref{eq_property_valid}, violates Condition~\eqref{eq_property_onesum}, and has the persistency property.
\end{proposition}

\begin{proof}
  It is easy to see that $R^{\text{(C)}}$ satisfies Conditions~\eqref{eq_property_facet} and~\eqref{eq_property_valid}.
  It violates Property~\eqref{eq_property_onesum} since for the graph $G = K_2 \onesum{}{} K_3$, the clique inequality for the $3$-clique is part of $R^{\text{(C)}}(K_3)$ but not valid for $R^{\text{(C)}}(G)$ since the formulation for $G$ only consists of edge- and nonnegativity constraints.
  By Proposition~\ref{thm_persistency}, $R^{\text{(C)}}$ has the persistency property.
\end{proof}

\section{Concluding remarks}
\label{sec_conclusion}

We have shown that persistency is an exceptional property for linear programming stable set relaxations.
Apart from studying nonlinear relaxations (such as those stemming from semidefinite relaxations), it is natural to ask whether this is also the case for other polytopes for which persistency was established.

The most interesting candidate is certainly the unconstrained quadratic binary programming problem, which is equivalent to the maximum cut problem.
The standard \emph{McCormick relaxation} also has the (weak and strong) persistency property~\cite{HammerHS84} with respect to the original variables (not corresponding to linearized products).
In fact, there is a strong relationship to the stable set problem as both problems can be easily reduced to each other.
Polyhedrally speaking, each polytope (relaxation or integer hull) can be obtained as a face of the polytope of the other problem, potentially after removing constraints that are redundant for a given objective vector~\cite{HammerHS84}.
Although this was used to show that the McCormick relaxation partially has the persistency property, the non-existence of the property for tighter relaxations does not carry over in a straight-forward manner.
Thus, we leave the resolution of this question as an open problem.

\paragraph{Acknowledgements.}
We are grateful to four anonymous reviewers whose comments led to improvements of this manuscript.
We also want to thank Yuri Faenza for asking for necessity of the properties~\eqref{eq_property_facet}--\eqref{eq_property_onesum}.

\bibliographystyle{plain}
\bibliography{persistency}

\appendix

\section{Deferred proofs}
\label{appendix_proofs}

We repeat the statements of \cref{thm_face_vertex} and \cref{thm_polytope_value_function} and provide their proofs.

\rsfacevertex*
\begin{proof}
  Let $c' \in \R^n$ be such that $\dim(\opt{Q}{c'}) < \dim(\opt{P}{c'})$ holds, and among those, such that $\dim(\opt{Q}{c'})$ is minimum.
  Clearly, $c'$ is well-defined since $c' \coloneqq c$ satisfies the conditions.

  Assume, for the sake of contradiction, that $\dim(\opt{Q}{c'}) > 0$.
  Let $F \coloneqq \opt{P}{c'}$ and $G \coloneqq \opt{Q}{c'}$.
  Let $F_1, F_2, \dotsc, F_k$ be the facets of $F$.
  By $n(F,F_i)$ we denote the set of vectors $w \in \R^n$ such that $\opt{F}{w} \supseteq F_i$.
  Since $F$ is a polytope, $\bigcup_{i \in \{1,2,\dotsc,k\}} n(F,F_i)$ contains a basis $U$ of $\R^n$.
  Moreover, not all vectors $u \in U$ can lie in $\aff(G)^{\perp}$, the orthogonal complement of $\aff(G)$, since then $\aff(G)^{\perp} = \R^n$ would hold, contradicting $\dim(G) > 0$.
  Let $u \in U \setminus \aff(G)^{\perp}$.

  Now, for a sufficiently small $\varepsilon > 0$, $\opt{P}{c' + \varepsilon u} \supseteq F_i$ for some $i \in \{1,2,\dotsc,k\}$, and $\opt{Q}{c' + \varepsilon u}$ is a proper face of $G$.
  Thus, $c' + \varepsilon u$ satisfies the requirements at the beginning of the proof.
  However, $\dim(\opt{Q}{c' + \varepsilon u}) < \dim(G)$ contradicts the minimality assumption, which concludes the proof.
\end{proof}

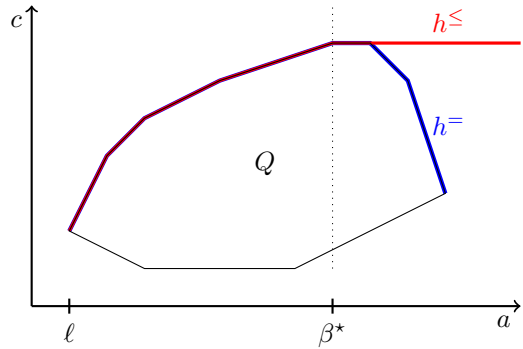
\begin{figure}[htb]
  \centering
  \begin{tikzpicture}
    \draw[thick,->] (-1.5,-0.5) -- (5,-0.5) node[anchor=north east] {$a$};
    \draw[thick,->] (-1.5,-0.5) -- (-1.5,3.5) node[anchor=north east] {$c$};

    \node at (1.6,1.4) {$Q$};

    \draw[thick] (-1,-0.4) -- (-1,-0.6) node[anchor=north] {$\ell$};
    \draw[thick] (2.5,-0.4) -- (2.5,-0.6) node[anchor=north] {$\beta^\star$};

    \draw[blue, ultra thick] (-1,0.5) -- (-0.5,1.5) -- (0,2) -- (1,2.5) -- (2.5,3) -- (3,3) -- (3.5,2.5) -- (4,1);
    \draw[red, very thick] (-1,0.5) -- (-0.5,1.5) -- (0,2) -- (1,2.5) -- (2.5,3) -- (5,3);
    \draw (0,0) -- (-1,0.5) -- (-0.5,1.5) -- (0,2) -- (1,2.5) -- (2.5,3) -- (3,3) -- (3.5,2.5) -- (4,1) -- (2,0) -- cycle;

    \draw[dotted] (2.5,0) -- (2.5,3.5);

    \node[red,anchor=west] at (3.7,3.3) {$h^{\leq}$};
    \node[blue,anchor=west] at (3.7,1.9) {$h^=$};
  \end{tikzpicture}
  \caption{Illustration of \cref{thm_polytope_value_function}. The graph of $h^{\leq}$ is highlighted in red, while that of $h^=$ is highlighted in blue.}
  \label{fig_polytope_value_function}
\end{figure}

\rspolytopevaluefunction*
\begin{proof}
  Let $Q \coloneqq \setdef{ \left(\begin{smallmatrix} y_1 \\ y_2 \end{smallmatrix}\right) }[ \exists x \in P :  \transpose{a}x = y_1,~ \transpose{c}x = y_2 ] \subseteq \R^2$ be the projection of $P$ along $a$ and $c$.
  By construction, $h^{\leq}(\beta) = \max \setdef{ y_2 }[ y \in Q,~ y_1 \leq \beta ]$ holds.
  Considering that $Q$ is a polytope of dimension at most $2$, the claimed properties of $h^{\leq}$ and $h^=$ are obvious (see \cref{fig_polytope_value_function}).
\end{proof}

\end{document}